\documentclass[letterpaper, 10pt, conference]{ieeeconf}
\IEEEoverridecommandlockouts
\usepackage{syntax}
\usepackage{url}

\usepackage{graphicx}
\usepackage{amssymb}

\usepackage{amsfonts}
\usepackage[english]{babel}
\usepackage[bottom]{footmisc}
\usepackage{wrapfig}
\usepackage[noend]{algpseudocode}
\usepackage{algorithm}
\usepackage{algpseudocode}
\usepackage{amsthm}
\usepackage{siunitx}
\usepackage{array}
\usepackage{verbatim}
\usepackage{subcaption}
\usepackage{listings}
\usepackage{amsmath}
\usepackage{adjustbox}
\usepackage{csquotes}
\usepackage{microtype}
\usepackage{graphicx}
\usepackage{amsmath}
\usepackage{todonotes}
\usepackage{fdsymbol}
\usepackage{xcolor}

\usepackage{amsfonts}

\newtheorem{theorem}{Theorem}
\newtheorem{remark}{Remark}
\newtheorem{definition}{Definition}

\newtheorem{proposition}{Proposition}

\let\emptyset\varnothing

\begin{document}

% \title{ PID Control of Biochemical Reaction Networks
% %\thanks{Identify applicable funding agency here. If none, delete this.}
% }

% \author{\IEEEauthorblockN{Max Whitby}
% \IEEEauthorblockA{\textit{Department of Computer Science} \\
% \textit{University of Oxford}\\
% max.whitby@cs.ox.ac.uk}
% \and
% \IEEEauthorblockN{Luca Cardelli}
% \IEEEauthorblockA{\textit{Department of Computer Science} \\
% \textit{University of Oxford}\\
% luca.cardelli@cs.ox.ac.uk}
% \and
% \IEEEauthorblockN{Marta Kwiatkowska}
% \IEEEauthorblockA{\textit{Department of Computer Science} \\
% \textit{University of Oxford}\\
% marta.kwiatkowska@cs.ox.ac.uk}
% \and
% \IEEEauthorblockN{Luca Laurenti}
% \IEEEauthorblockA{\textit{Department of Computer Science} \\
% \textit{University of Oxford}\\
% luca.laurenti@cs.ox.ac.uk}
% \and
% \IEEEauthorblockN{Mirco Tribastone}
% \IEEEauthorblockA{
% \textit{IMT School for Advanced Studies}\\
% Lucca, Italy\\
% mirco.tribastone@imtlucca.it}
% \and
% \IEEEauthorblockN{Max Tschaikowski}
% \IEEEauthorblockA{\textit{Department of Computer Engineering} \\
% \textit{Vienna University of Technology}\\
% max.tschaikowski@tuwien.ac.at}
% }

\title{\LARGE \bf
PID Control of Biochemical Reaction Networks
}

\author{Max Whitby$^{1}$, Luca Cardelli$^{1}$, Marta Kwiatkowska$^{1}$, Luca Laurenti$^{1}$, Mirco Tribastone$^{2}$, Max Tschaikowski$^{3}$
% <-this % stops a space
%\thanks{*This work was not supported by any organization}% <-this % stops a space
\thanks{$^{1}$ Department of Computer Science, University of Oxford, UK}
\thanks{$^{2}$ IMT School for Advanced Studies, Lucca, Italy}%
\thanks{$^{3}$ Department of Computer Engineering, TU Wien, Austria}%
}

\maketitle

\begin{abstract}
%Contributions: differentiator, controlability, Dual rail to single rail, gene expression
%Control of biochemical processes has application in metabolic engineering and synthetic biology. Current approaches do not address the fact that systems are naturally occuring systems are single rail. \LL{NOt totally  true. there are many papers that do that, just they are  ad-hoc}Current controllers are also PI 
%\LL{I tried to revise a bit the abstract and intro to be clearer on what are the contributions of the paper. We still need to polish them.}\IMT{Looks fine to me}
Principles of feedback control have been shown to naturally arise in biological systems and %have been 
successfully applied %with success 
to build synthetic circuits.
In this work we consider Biochemical Reaction Networks (CRNs) as a paradigm for modelling biochemical systems and provide the first implementation of a derivative component in CRNs. That is,
given an input signal represented by the concentration level of some species, we build a CRN that produces as output the concentration of two species whose difference is the derivative of the input signal. %\LC{Strictly we do not produce "a species" but the difference of two species: how can we finesse this?}\LL{I tried to modify this without going too much in the detail, please check.}\MK{But is this just dual rail? If so it is better to say it.} \LC{Seems fine now, we probably do not want to use "dual-rail" (without explanation) in the abstract. Dual-rail usually means encoding a boolean signal on two boolean lines, not a difference like here. I have been using "differential representation" but that's a bit confusing when we also talk about derivatives. Maybe just "difference representation"?} 
By relying on this component, we present a CRN implementation of a feedback control loop with Proportional-Integral-Derivative (PID) controller and apply the resulting control architecture to regulate the protein expression in a \emph{microRNA} regulated gene expression model.

%Biochemical systems can often be prescribed by Chemical Recomponent Networks and they are used as schematics for chemical engineering. Previous results for feedback control with Chemical Recomponent Networks rely on Proportional-Integral(PI) feedback. This can lead to instability and oscillations around a fixed point. We present the first Chemical Recomponent Network implementation of a Proportional-Integral-Derivative (PID) controller. Our derivative component is in itself also original. We show how the proposed circuit can be used to control naturally occurring gene expression systems and compare it to similar results of a PI controller.
\end{abstract}

%\section{Jottings}
%\begin{itemize}
    % \item \IMT{Is it not more appropriate for the title to be ``Biochemical implementation of PID control'', or something like that? Because, strictly speaking, we could use it for any control system} \LC{The title reflects exactly the original intent, the proposed new title reflects the opposite intent... The topic of controlling biochemical reaction networks is hot in synthetic biology. The topic of chemically controlling arbitrary things, I don't know. We use a chemical controller because we want to control chemical things, that's the motivation.}
    %\item \TUW{We should stress somewhere that we do not have to forward the outputs of our I and D blocks to P blocks because our I and D blocks have ``integrated'' multipliers. (We mention the multipliers in the main figure but this is probably not enough).}

%    \item \TUW{In general, the output and reference signals are not one dimensional. Hence, we may mention in the caption of the main figure that all blocks have to be interpreted as vectors. We say, however, it also in the main text above Definition 1.}
%    \item \TUW{We say several times that we provide a CRN encoding of a \emph{negative} PID feedback controller. I think the main message is that we provide a PID feedback controller.}\MW{Agreed}
%    \item \TUW{To be consistent, we are now using derivative block instead of differential block; also, we use block instead of box.}
%\end{itemize}
%%%%%%%%%%%%%%%%%%%%%%%%%%%%%%%%%%%%%%%%%%%%%%%%%%%%%%%%%%%%%%%%%%%%%%%%%%%%%%%%
\section{INTRODUCTION}

Biochemical Reaction Networks (CRNs) are a widely used formalism to describe biochemical systems  \cite{erdi1989mathematical}. More recently, they have also been employed as a formal programming language for synthetic circuits made of DNA~\cite{Soloveichik5393, chen2013programmable}.  Due to the numerous potential applications, ranging from smart therapeutics to biosensors, the construction of CRNs that exhibit prescribed dynamics is a major goal of synthetic biology.
However, achieving a desired behaviour by designing a CRN is difficult due to the  complexity of such systems and limited knowledge of their dynamics~\cite{cardelli2017syntax,chiu2015synthesizing}.

Negative feedback and Proportional-Integral-Derivative (PID) control are widely used in engineering to control the dynamics of a system due to their ability to achieve accurate set-point tracking and robustness to disturbances, even with only partial knowledge of the system. Because of these properties, such mechanisms have also been applied with success in the construction of synthetic bio-molecular systems~\cite{del2016control,kelly2018synthetic}. Moreover, molecular implementation of control systems has been shown to naturally occur in living organisms \cite{zhang2012design,dunlop2010model,yi2000robust}.
For example, integral control occurs in \emph{E.coli} chemotaxis \cite{barkai1997robustness,briat2016antithetic}, while \emph{CheY} proteins regulate the bacteria's tumbling frequency by implementing a derivative control \cite{alon1998response}. %\LL{Be careful to not replace "control" with "component" in the above sentence.}
%\LL{Max, check if this accurate with what they say in the references. I did not read the paper about the derivative control. So, let's be accurate.}
%step-changes in chemoattractant concentration, maintaining the system’s sensitivity to new concentration
%changes.  Indeed this mechanism also yields a naturally occuring biochemical implementation of a derivative action.
%\LL{Next sentence should be improved.}
As a consequence, in view of the potential applications, CRN designs that implement control mechanisms are %receiving increasing interest
sought for~\cite{briat2016antithetic,del2016control}.
CRNs implementing proportional and integral control have been proposed~\cite{oishi2011biomolecular,briat2016antithetic}.
%In \cite{briat2016antithetic} the \emph{antithetic integral} feedback molecular motif has been presented, which is a CRN that has been shown to guarantee network ergodicity and zero steady state error, while in \cite{oishi2011biomolecular} a CRN implementation of proportional and integral control has been given. However
However, a biochemical implementation of a full PID control is still missing due to the lack of a CRN implementing the derivative component. %\IMT{the derivative operator?}.

In this work we first present a CRN implementation of a derivative component. %\IMT{always use the same name: module/operator/function}\LC{/action/component}. 
That is, we provide a CRN such that, given an input signal represented by the concentration level of some species, the output is the concentration of two species %\LC{same comment as in Abstract} 
whose difference gives the derivative of the input signal. %Being molecular concentrations strictly positive by definition, we consider a \emph{dual rail encoding} for the output of our derivative module, where the time evolution of the output is given by the difference between the time evolution of two species \IMT{we need to cite Oishi and Klavins for this enconding, or earlier papers? Also, the point about dual-rail encoding does not need to be highlighted in the introduction in my opinion.}.
We use this as a building block for a PID controller, and show how negative feedback with PID controller can be implemented in CRNs.
We show the effectiveness of this architecture on a microRNA regulated gene expression example \cite{schmiedel2015microrna,laurenti2018molecular}, where we control the time evolution of a protein by acting on the expression of mRNA and microRNA.

In summary, we make the following contributions:
\begin{itemize}
\item We present a CRN that computes the derivative of an input signal and prove its asymptotic correctness.
\item We extend the correctness results of CRN encodings~\cite{oishi2011biomolecular} of proportional and integral signals to the case of nonlinear dynamics. 
%\LL{Added This} We extend the proof of correctness for the CRNs implementing proportional and integral functions presented in~\cite{oishi2011biomolecular} to non-linear dynamics.
\item We provide for an arbitrary CRN plant a CRN encoding of the PID feedback controller.
\item We show the effectiveness of our control architecture on a microRNA regulated gene expression model.
\end{itemize}

%\LL{This can be reduced and  merged somewhere in the intro, if needed.} \IMT{I think that it may even be removed for CDC.}\LL{I agree}
%Differential control acts as a predictor of future error, reducing the accumulated integral error. Dampening the error can be important when the dynamics of the plant is such that a simpler PI controller would lead to overshooting and oscillation around the setpoint. That is usually due to second order dynamics, such as when controlling the position of a mass through a force, or in biology when controlling the concentration of a protein by gene activation through mRNA production. In both these cases the controlled variable is two integrations away from the controlling signal. %\MWrev{In this work we demonstrate a full PID controller with a new derivative action, motivated by the need for error reduction in gene expression circuits.} %\LL{They are also much more sensible to noise. So, we may need to add low pass filters \cite{laurenti2018molecular} and do some stochastic simulations.}

\section{Biochemical Reaction Networks}

In this section we provide some background about the deterministic mass-action semantics of a CRN based on the reaction-rate equations. Then we review the notion of dual rail encoding for a species. Finally we fix a graphical representation of CRNs that will be used throughout the paper. 

\subsection{Deterministic Mass-action semantics}
A CRN $\mathcal{C}=(\cal S,\mathcal{R})$ %\IMT{why $\Lambda$ and not $\cal S$, which seems more natural?} 
is a pair of finite sets, where $\cal S$ is an ordered set of species, $|\cal S|$ denotes its size, and $\mathcal{R}$ is an ordered set of reactions. Species in $\cal S$ interact according to the reactions in $\mathcal{R}$. A reaction $\tau \in \mathcal{R}$ is a triple $\tau=(r_{\tau},p_{\tau},k_{\tau})$, where $r_{\tau} \in  \mathbb{N}^{|\cal S|}$ is the \emph{reactant complex}, 
$p_{\tau} \in  \mathbb{N}^{|\cal S|}$ is the \emph{product complex} and $k_{\tau} \in \mathbb{R}_{>0} $ is the coefficient associated with the rate of the reaction. Complexes $r_{\tau}$ and $p_{\tau}$ represent the stoichiometry of reactants and products.
We denote the $i$-th component of complex $r_\tau$ by $r_{\tau,i}$;  the zero complex is denoted by $\emptyset$. Given a CRN with species set $\cal S$ $ = \{ A, B, C \}$, a reaction $(  [1,1,0],[0,0,2],k_1 )$ will be denoted by $A + B \, \rightarrow^{k_1}  \,2C$. 
The \emph{state change} associated to $\tau$ is defined by $\upsilon_{\tau}=p_{\tau} - r_{\tau}$.  For example, the state change of the reaction above is $[-1,-1,2]$. 

We consider the deterministic interpretation of a CRN based on the well-known reaction-rate equations with mass-action kinetics. Given a CRN $\mathcal{C}=(\cal S,\mathcal{R})$ and an initial condition $x_0 \in \mathbb{R}^{|\cal S|}_{\geq 0}$ representing the initial concentration of each species, the time course of the concentrations can be described as the solution of an initial value problem with the following system of ODEs
\begin{equation}
 \partial_t x(t)  =\sum_{(r_{\tau},p_{\tau},k_{\tau}) \in \mathcal{R}} \upsilon_{r} k_{\tau}\prod_{i=1}^{|\cal S|}{x_i(t)}^{r_{\tau,i}},
\label{eq:ODEDetermnisticSemanticsCRN}
\end{equation}
and initial condition $x(0) = x_0$. For a species $A \in \cal S$ we denote by $x_A(t)$ the concentration of $A$ at time $t$. 

In this paper we synthesise PID controllers with mass-action kinetics for CRNs. %\TUW{Modified} 
We will also assume that the plant is represented by a mass-action CRN, although our results carry over to plants given in terms of smooth control systems.

%consisting of CRNs with arbitrary, but smooth, kinetics.

%\LL{What do we mean with control systems in here? The plant? If so, it is not super clear, as the control system is generally given by the mechanisms used to control the plant and not by the plant itself. But, if you think it is clear, I am fine with that.}

%$\Phi : \mathbb{R}_{\geq 0} \rightarrow \mathbb{R}^{|\Lambda|} $ describes the concentration of the species over time, therefore $\Phi(t) \in \mathbb{R}^{|\Lambda|}_{\geq 0}$ is the vector of the species concentrations at time $t.$ 
% \LL{Need to assume there exists a steady state solution.} 
%\LL{Here add definition of Dual Rail Encoding}

% \IMT{The paragraph mentions dual rail encoding, but the first sentences discuss I/O blocks, which seems a separate issue. Also, I/O blocks do not seem to be formalized, I wonder if we need to do this}\MW{We had this discussion but felt that in order to formalise it would take a lot of space. I've split it into two areas. Feel free to formalise slightly more. Changed species names also.}

%\subsection{CRNS as I/O blocks}
%We envisage CRNs in terms of I/O blocks where species within the block can have the label input/output. Internal species specific to that block have no chemical interactions with species outside of the block unless labelled input/output. I/O blocks can be composed where the output species of one block can be the input species of another. \MW{Really good restructure. Ive now added new image (caption not updated)}

\subsection{Dual Rail Encoding}
%\IMT{Rewrote this} 
The plant is a CRN, hence its output is given by non-negative solutions. However, the PID controller will involve quantities that are negative such as the error, i.e., difference between the set-point and the output, as well as its derivative. In order to handle this we will use the so-called dual rail encoding~\cite{oishi2011biomolecular}, by which a  a signal is decomposed into a ``positive'' and ``negative'' species component whilst preserving a law of mass action kinetics such that each individual species concentrations cannot be negative. Specifically, for a signal $A$ we denote the two distinct component species by $A^+$ and $A^-$, representing the positive and negative signals, respectively. %\TUW{Modified} 
Owning to the fact that the plant is described by a mass-action CRN, we wish to point out that its non-negative output can (and is) captured by single rail encoding. 

%\IMT{I do not fully understand the next sentence} All components of a PID controller are designed to handle such dual real signals, except where they interface with the (single rail) output of the CRN plant, which is non-negative. 
%\IMT{If there is space, can we add a minimal example of dual-rail encoding at work?} %\TUW{Would remove next sentence, normalisation not clear for CDC community here.} As part of the normalisation to single rail, we often employ reactions of the form $A^+ + A^- \rightarrow \emptyset$ that do not affect the represented difference.

\subsection{Graphical Representation of CRNs}

A CRN can be represented as a labelled directed bipartite graph according to the usual Petri net representation with species and reaction nodes~\cite{cardelli2014morphisms}. A reaction node is labelled with a rate coefficient. There is an edge from a species node to a reaction node if the species is in the reactant complex, with label equal to its multiplicity; similarly, an edge from a reaction node to a species node indicates the presence of a species in the product complex. For example, we have the following representation for the reaction $2A + B \xrightarrow{k} B + 3A$:
\begin{minipage}{\linewidth}
       \centering
       \includegraphics[width=7cm]{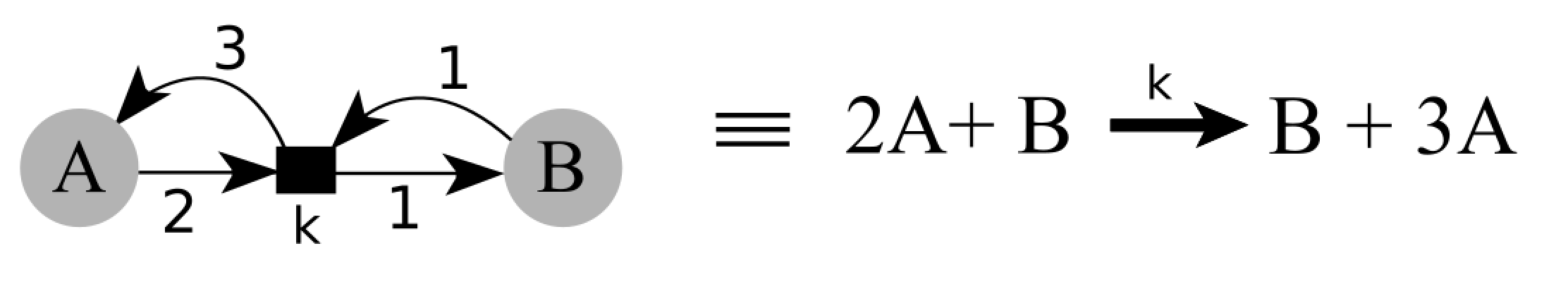}
       %\captionof{figure}{Rules for equivalence of our diagrammatic notation with petri nets}
       %\label{petriequivalence}
\end{minipage}

\begin{figure}[t]
       \centering
       \includegraphics[width=8.5cm]{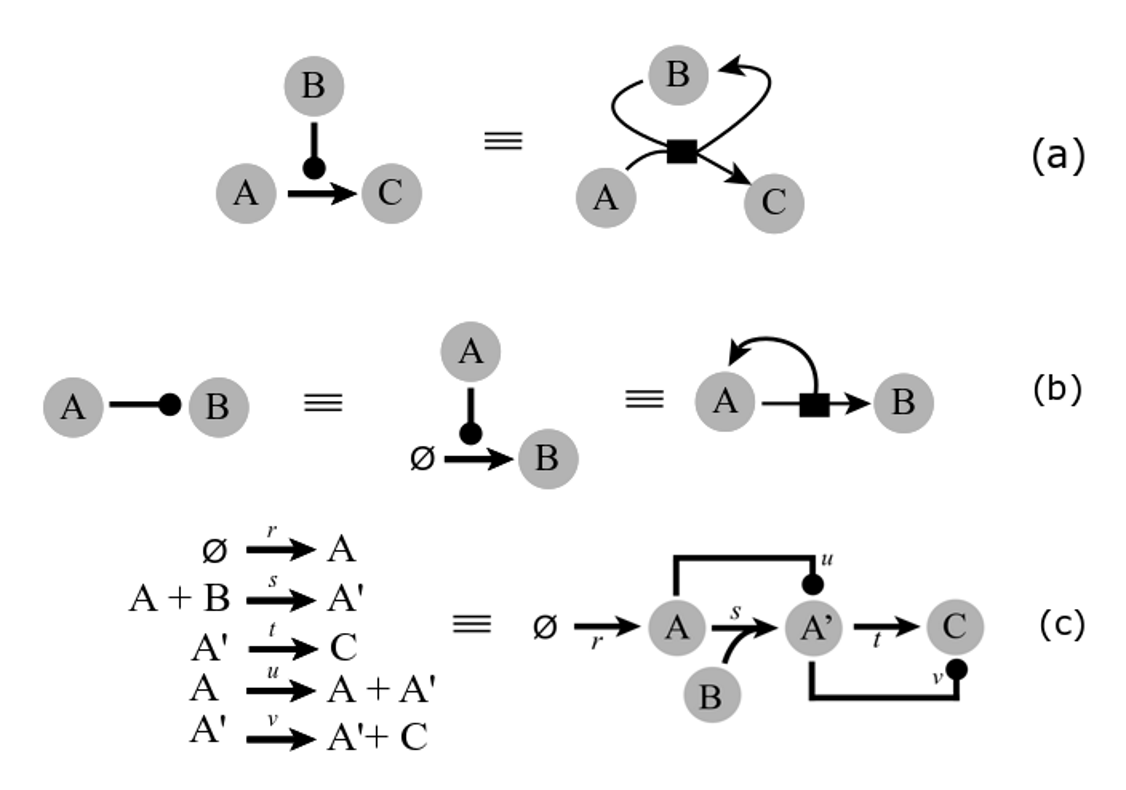}
       \caption{Short-hand CRN graphical notation. (a) A catalytic bi-molecular reaction  $A + B \rightarrow B + C$ as an equivalent Petri net; (b) a catalytic uni-molecular reaction $A \rightarrow A + B$; (c) a sample CRN depicted using the short-hand notation.
       %\LC{Part (c) needs fixing but I could not find the .svg for this figure. The fonts are wrong on the left, and the right part is squashed horizontally}
       }\label{petriequivalence}
\end{figure}
Throughout the rest of the paper, for ease of presentation we do not draw labels on edges if the related complex multiplicity is~1. We also remove the black box representing reaction nodes to reduce clutter. 
Finally, we introduce a short-hand notation for recurring reaction patterns as shown in Figure~\ref{petriequivalence}, where each arc is either a pointed arrow ($\uparrow$) or a  rounded arrow ($\upblackspoon$), with the source represented by the flat edge and the target represented by the arrow head. A pointed arrow represents a non-catalytic reaction and a rounded arrow represents a catalytic reaction.    

\section{CRN implementation of the PID controller}\label{CRNImplementation}

%\subsection{Composition}
%\LL{Input and output are species of the CRN, not of the circuit. What you compose is CRNs.}
%We represent a circuit with input species i.e. $A^{+}$, outputs i.e. $B^{-}$ and sometimes some internal species i.e. $C^{+}$. If two circuits are composed in series then the output species of the first circuit becomes the input species to the second circuit. Two species composed in parallel assumes that species are separate nad their is no interaction between modules. As we are using a dual rail controller a positive species is represented by a $+$ symbol and similarly a negative by $-$. 

%\LL{Here a general and concise description of PID controller. put equations, as this is a theoretical hournal. }

%\subsection{Composition of Controller}
%Our control loop consists of a reference signal $S$, a controller $C$ and a plant $P$. We use reference signals given in 

\begin{figure*}[!ht]
\centering
\begin{minipage}[l]{2.0\columnwidth}
\centering
\includegraphics[width=16cm]{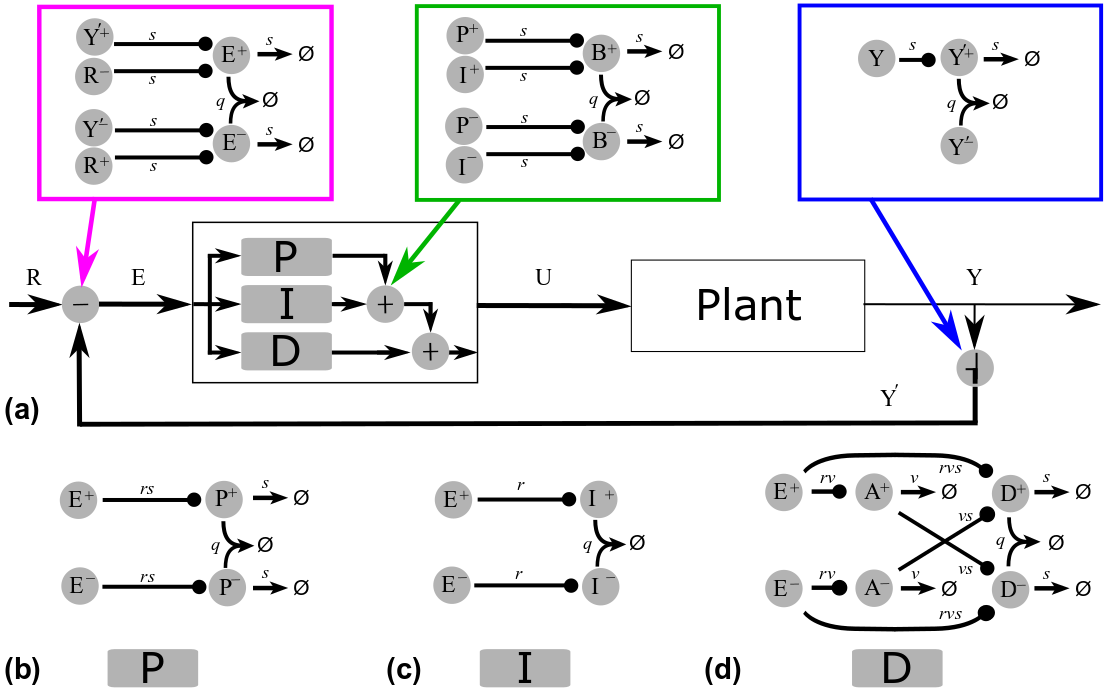}
\caption{
We present our feedback loop (a) which takes a smooth reference signal $R$ (in dual rail in its most general case) and, along with the feedback $Y$, produces an error $E$. Signal $E$ is obtained by the subtraction block (in magenta). Thick arrows imply dual rail whereas the thin arrow $Y$ implies single rail. $E$ is fed to the controller, the chemical composition of which is described in (b),(c),(d). The proportional and integral blocks (b),(c) are taken from \cite{oishi2011biomolecular}. The proportional block adjusts the input $x_A$ as $r x_A$ for some multiplier $r \geq 0$. The integral block takes an input $x_A$ and produces $r \int^{t}_0 x_A(\tau) d\tau$ for some multiplier $r \geq 0$. 
%We propose a novel derivative component (d). 
Instead, the novel derivative block (d) takes an input $x_A$ and produces the output $r \partial_t x_A$, where $r \geq 0$ is a multiplier. The foregoing blocks are summed by the addition block (in green), yielding a control signal $U$ which steers the plant by the CRN encoding presented in Section~\ref{CRNImplementation}. As a result, the plant produces a signal $Y$ which is converted to a dual rail signal $Y'$ by the dual rail converter block (in blue). The presence of a multiplier in each block allows one to adjust the weights of each block. In particular, parameter values $r,s,v,q$ are block dependent in general.}
\label{PIDfull}
\end{minipage}
\end{figure*}

 We introduce the CRN implementation of the proportional, integral, and derivative components of a PID controller. We describe them as blocks where the input species are $E^\pm$ (which will indicate the dual-rail error signal between the species representing the set-point and the plant output). The output of the PID controller is denoted by $U^\pm$. The proportional and integral components have been already introduced for linear control systems~\cite{oishi2011biomolecular}. Here we prove their correctness in the presence of non-linearity. In addition we detail the CRN implementation of the derivative block, which is a novel contribution to the best of our knowledge. 
 %We report all three circuits and give short proofs of their correctness. We further show how the proportional, integral, and derivative blocks can be employed to control naturally occurring biological signals. This is then used to provide a CRN realisation of a feedback control loop with a PID controller.

As shown in Figure \ref{PIDfull}, as in a classic feedback loop, the signals synthesised by the PID controller act on the CRN plant, whose output is measured, and sent back as input of the PID controller after comparison with the reference signal. The output of the plant is always given by a species $Y$. 
The objective of the control is to have $Y$ follow the reference signal. 

To this end, we let $(\mathcal{S}_\Sigma,\mathcal{R}_\Sigma)$ denote the mass-action CRN representing the plant and construct a CRN encoding of a PID feedback law as indicated in Figure~\ref{PIDfull}. %\IMT{To Max T: This $n$ comes from nowhere and I don't recall your explanation of it. I think that we need more explanation here.} \TUW{Modified.} 
Given that output and control signals are vectors in general, the blocks of Figure~\ref{PIDfull} are multidimensional components in general. In particular, assuming that $n \geq 1$ denotes the dimension of the output and control vector, the CRN encoding of the PID feedback law is given by interconnected
\begin{itemize}
    \item subtraction blocks $(\mathcal{S}^S_i,\mathcal{R}^S_i)_{1 \leq i \leq n}$
    \item addition blocks $(\mathcal{S}^A_i,\mathcal{R}^A_i)_{1 \leq i \leq n}$ 
    \item proportional blocks $(\mathcal{S}^P_i,\mathcal{R}^P_i)_{1 \leq i \leq n}$
    \item integral blocks $(\mathcal{S}^I_i,\mathcal{R}^I_i)_{1 \leq i \leq n}$
    \item derivative blocks $(\mathcal{S}^D_i,\mathcal{R}^D_i)_{1 \leq i \leq n}$
    \item dual rail converter blocks $(\mathcal{S}^C_i,\mathcal{R}^C_i)_{1 \leq i \leq n}$.\footnote{Please note that imposing the presence of a $P$, $D$ and $I$ block for every coordinate $1 \leq i \leq n$ is without loss of generality because a block can be removed by setting its multiplier $r$ to zero.}
\end{itemize}
With this, the overall CRN is given by
\begin{align}\label{eq:fullsys}
(\mathcal{S}_\Sigma,\mathcal{R}_\Sigma) \cup (\mathcal{S}_F,\mathcal{R}_F) ,
\end{align}
where the feedback law CRN is defined by
\begin{align*}
(\mathcal{S}_F,\mathcal{R}_F) & = \bigcup_{i=1}^n \bigcup_{X \in \mathcal{X}} (\mathcal{S}^X_i,\mathcal{R}^X_i) ,  & 
\mathcal{X} & = \{ S,A,P,I,D,C \} .
\end{align*}

In what follows next, we address the correctness of each block type.

\subsection{Proportional, Addition, Subtraction and Dual Rail Converter Blocks}

We begin by presenting the proportional block which computes an output signal that is proportional to the input signal. 

% \TUW{Read next comment} \LC{I do not understand, what is the "i-th" proportional block? Aren't we defining "a" generic proportional block in this section, exactly like Definition 1 already says? If you need some notation, put it inside the theorem's proof.}

% \textbf{\emph{Notation.}} For the benefit of presentation, we shall suppress the index $1 \leq i \leq n$ of the current block in question. For instance, in the case we consider the $i$-th proportional block, we write $E^+$, $x_{E^+}$ and $s$ instead of $E^+_i$, $x_{E^+_i}$ and $s_{P_i}$, respectively.

%This component is widely used in control systems to react proportionally to the current error with respect to the reference signal. %Proportional blocks produce linear combinations of their inputs. \MWrev{Our proportional blocks take one input and produce one output. They can be composed such that there are several inputs and several outputs.} They are used for combining signals from the integral and differentiator within our PID. A single-input block takes an input signal A(t) and produces an output signal $B(t) = kA(t)$ where $k\in \mathbb{R}$ is the gain. A summation block takes input signals $\{A_i(t)\}^n_{i=1}$ and produces an output $B(t) = \sum^n_{i=1}A_i(t)$. The following CRN implements both of these functions for the dual rail representation of signals

%The following PI components of our PID are given in \cite{klavins} but we list the definitions for completeness. We expand on Klavins result for PI controller by showing completeness. 

%\LL{Need to write a clearer introduction paragraph}

\begin{definition}[Proportional block~\cite{oishi2011biomolecular}]\label{Def:PropBlaock}
%Gain and summation blocks produce linear combinations of their inputs. They are used for combining signals from the integral and differentiator within our PID. A gain signal takes an input signal A(t) and produces an output signal $B(t) = kA(t)$ where $k\in \mathbb{R}$ is a rate. A summation block takes an input $\{A_i(t)\}^n_{i=1}$ and produces an output $B(t) = \sum^n_{i=1}A_i(t)$. The following CRN implements both of these functions
For input species $E^{+},E^-$, output species $P^+,P^-$,  parameters $s, q \in\mathbb{R}_{> 0}$, and the multiplier $r \in \mathbb{R}_{\geq0}$, the \emph{proportional block} is a CRN composed by the following reactions 
%\TUW{The rates below are important, do not change them, otherwise the correctness proof will fail; adjust the pictures, instead. CRUCIALLY: do not replace $q$ with $s$ or $rs$! This applies to all blocks.} 
\begin{align*}
    &E^+  \overset{rs}{\rightarrow} E^+ + P^+  
    &&E^-  \overset{rs}{\rightarrow} E^- + P^- \\
    &P^+ + P^-  \overset{q}{\rightarrow}\emptyset
    &&P^+  \overset{s}{\rightarrow}\emptyset \\
    &P^-  \overset{s}{\rightarrow}\emptyset
\end{align*}  
\end{definition}
%\LL{Need to be clear about how we call the ODE associated to a species. HEre I used some name.}
%\LC{Recast proposition in terms of $\{A^{+}_i,A^-_i\}^n_{i=1}$}

%\IMT{why plural? (i.e., blocks instead of block)} \TUW{Because the output and control are $n$-dimensional VECTORS in general}

\begin{theorem}\label{prop:proportional}
On any bounded time interval, the ODE system of~(\ref{eq:fullsys}) converges to an ODE system satisfying $x_{P^+_i} = r x_{E^+_i}$ and $x_{P^-_i} = r x_{E^-_i}$ for all $1 \leq i \leq n$ if $s \to \infty$ in all proportional blocks.
\end{theorem}

\begin{proof}[Proof (Sketch)]
Note that
\begin{align*}
    \partial_t x_{P^+_i} & = rs x_{E^+_i} - s x_{P^+_i} - q x_{Pr^+_i} x_{P^-_i} \\
    \partial_t x_{P^-_i} & = rs x_{E^-_i} - s x_{P^-_i} - q x_{P^+_i} x_{P^-_i}    
\end{align*}
This motivates to interpret all $x_{P^+_i}$ and $x_{P^-_i}$ as fast variables in the sense of Tikhonov's theorem~\cite[Section 8.2]{Verhulst2005}. To see that the requirements of the theorem are satisfied, we note that the fast ODE system (i.e., the one consisting of fast variables) admits, for any fixed vector of slow variables (i.e., all variables that are not fast), exactly one possible equilibrium point. Moreover, it is an asymptotically stable equilibrium of the fast ODE system. We finish the proof by noting that smooth exogenous reference signals can be captured because Tikhonov's theorem applies to non-autonomous smoooth ODE systems.
\end{proof}

\begin{remark}
Theorem~\ref{prop:proportional} extends the result of~\cite{oishi2011biomolecular} to nonlinear control systems. The same holds true for the other blocks of this section.
\end{remark}

\begin{remark}
The proof of Theorem~\ref{prop:proportional} reveals that reaction 
$P^+ + P^- \overset{q}{\rightarrow} \emptyset$ is not strictly needed to ensure correctness. However, this reaction precludes $P^+$ and $P^-$ from attaining excessively large values (recall that $s$ is large), thus reducing the impact of numerical errors.
\end{remark}

The correctness of proportional, subtraction and converter blocks introduced next and depicted in Figure~\ref{PIDfull} is shown similarly to  Theorem~\ref{prop:proportional}. 

More specifically, the addition block is given by.

\begin{definition}[Addition block~\cite{oishi2011biomolecular}]
For input species $P^{+},I^+$ and $P^{-},I^-$, output species $E^+,E^-$, and parameters $s, q \in\mathbb{R}_{> 0}$ the \emph{addition block} is a CRN composed by the following reactions 
\begin{align*}
   P^+ & \overset{s}{\rightarrow} P^+ + B^+ &
   I^+ & \overset{s}{\rightarrow} I^+ + B^+  \\
   P^- & \overset{s}{\rightarrow} P^- + B^-  &
   I^- & \overset{s}{\rightarrow} I^- + B^-  \\
   B^+ + B^- & \overset{q}{\rightarrow} \emptyset  &
   B^+ & \overset{s}{\rightarrow} \emptyset  &
   B^- & \overset{s}{\rightarrow} \emptyset  
\end{align*}
\end{definition}

The subtraction block, instead, is defined as follows.

\begin{definition}[Subtraction block~\cite{oishi2011biomolecular}]
For input species $Y'^{+},R^+$ and $Y'^{-},R'^-$, output species $E^+,E^-$, and parameters $s, q \in\mathbb{R}_{> 0}$ the \emph{subtraction block} is a CRN composed by the following reactions 
\begin{align*}
   Y'^+ & \overset{s}{\rightarrow} Y'^+ + E^+ &
   R^- & \overset{s}{\rightarrow} R^- + E^+  \\
   Y'^- & \overset{s}{\rightarrow} Y'^- + E^-  &
   R^+ & \overset{s}{\rightarrow} R^+ + E^-  \\
   E^+ + E^- & \overset{q}{\rightarrow} \emptyset  &
   E^+ & \overset{s}{\rightarrow} \emptyset  &
   E^- & \overset{s}{\rightarrow} \emptyset 
\end{align*}
\end{definition}

%Similarly, the dual rail converter block is defined as follows 
%The dual-rail converter block prepares $Y$, the non-negative output of the plant CRN, for being handled by the dual-rail encoding of the controller, thus transforming $Y$ into two distinct species  $Y'^+$ and $Y'^-$.

At last, the converter block is described by the following.

\begin{definition}[Dual rail converter block]
For input species $Y$, output species $Y'^+,Y'^-$, and parameters $s, q \in\mathbb{R}_{> 0}$ the \emph{single to dual rail converter  block} is a CRN composed by the following reactions 
\begin{align*}
   Y & \overset{s}{\rightarrow} Y + Y'^+ &
   Y'^+ & \overset{s}{\rightarrow} \emptyset &
   Y'^+ + Y'^- & \overset{q}{\rightarrow} \emptyset   
\end{align*}
\end{definition}

\subsection{Integral Block}

%\LC{Please make a subsection for Integral Block starting here and change the title of subsection A accordingly}\LL{Done}
The integral component computes a multiple of the integral of the input signal. This action is widely used in control systems due to its ability to collect past information about the error to be corrected. A CRN implementation of the integral component has been proposed in \cite{oishi2011biomolecular} and is reported in Definition \ref{Def:IntegrBlock}. 

\begin{definition}[Integral block~\cite{oishi2011biomolecular}]\label{Def:IntegrBlock}
For input species $E^{+},E^-$, output species $I^+,I^-$, some constant $q \in\mathbb{R}_{> 0}$ and multiplier $r \in \mathbb{R}_{\geq0}$, the integral block is given by the following CRN
\begin{align*}
   & E^+ \overset{r}{\rightarrow} E^+ + I^+ &
   & E^- \overset{r}{\rightarrow} E^- + I^- &
   & I^+ + I^-\overset{q}{\rightarrow}\emptyset
\end{align*}
\end{definition}

\begin{proposition}\label{prop:integral}
The integral blocks introduced in Definition~\ref{Def:IntegrBlock} are correct. More formally, for all $1 \leq i \leq n$, it holds that
\[
x_{I^+_i}(t) - x_{I^-_i}(t) = r \int_0^t x_{E^+_i}(\tau) d\tau - r \int_0^t x_{E^-_i}(\tau) d\tau
\]
\end{proposition}

\begin{proof}
Straightforward via differentiation (the values at $t = 0$ are chosen appropriately).
\end{proof}

Proposition~\ref{prop:integral} can be seen as a (straightforward) extension of the corresponding result in~\cite{oishi2011biomolecular} to nonlinear CRN plants.

\subsection{Derivative Block}
%\IMT{Use consistent name: we used derivative operator/component/etc in Intro, now it is derivative block. This text must be changed to be consistent with the new developments.}
The derivative block computes a multiple of the derivative of the input signal. This component is used in control systems to predict the future error given its current trend, and thus to help dampen oscillations introduced by P and I components.

Building a derivative module by chemical reactions is challenging because on-the-fly differentiation can only be done by comparing a signal at two time points, inherently requiring an approximation dependent on the time difference. This is resolved by the circuit in Figure \ref{PIDfull}(d), which handles dual rail input and output. 
Intuitively, the inputs $E^+$ and $E^-$ are sampled at two time points, $E^+$, $A^+$ and $E^-$, $A^-$, respectively, and a multiple of their difference is provided via $D^+$, $D^-$.
In Figure \ref{PIDfull}(d), the two reactions 
$E^+ \overset{rv}{\rightarrow} E^+ + A^+$ and $A^+ \overset{v}{\rightarrow} \emptyset$
cause $x_{A^+}$ to track $r x_{E^+}$ with a (slight) delay dependent on $v$. Intuitively, this yields $r x_{E^+} - x_{A^-} \approx \partial_t r x_{E^+}$. The symmetric two reactions similarly cause $x_{A^-}$ to track $r x_{E^-}$, which ensures that $r x_{E^-} - x_{A^-} \approx \partial_t r x_{E^-}$. The three reactions 
$E^+ \overset{rvs}{\rightarrow} E^+ + D^+$, 
$A^- \overset{vs}{\rightarrow} A^- + D^+$, and
$D^+ \overset{s}{\rightarrow} \emptyset$
cause $x_{D^+}$ to track $r v x_{E^+} + v x_{A^-}$ with delay dependent on $s$.
The symmetric three reactions similarly cause $x_{D^-}$ to track $r v x_{E^-} + v x_{A^+}$.
Thus $x_{D^+} - x_{D^-}$ tracks $v(r x_{E^+} + x_{A^-}) - v(r x_{E^-} + x_{A^+}) = v(r x_{E^+} - x_{A^+}) - v(r x_{E^-} - x_{A^-})$ with delay dependent on $s$. This and the above discussion allow us then to conclude that $x_{D^+} - x_{D^-} \approx r \partial_t (x_{E^+} - x_{E^-})$.
%\LC{Very good, but I put r outside in the last step: isn't that what we want to conclude? Also note that I removed the final "with a delay proportional to s". Not much point in emphasising this here, the original emphasis was actually about it being proportional to v, not to s, but with r in the mix it was now getting too obscure.} \TUW{I added $s$. In general, the main goal is imho to provide intuition why the derivative block approximates $r \partial_t x_E$. However, please feel free to modify.}

The next theorem formalises the above considerations using Tikhonov's theorem.

%Thus $x_{D^+} - x_{D^-}$ tracks $v(r x_{E^+} + x_{A^-}) - v(r x_{E^-} + x_{A^+}) = v(r x_{E^+} - r x_{E^-}) - v(x_{A^+} - x_{A^-})$, which represents the difference between the dual-rail input and its delayed copy proportionally to the delay $v$.

%\IMT{which one?}. \LC{I think it is v*s: they both should be big enough. But leave it like this: the theorem spells what each parameter does in the limit.} 

\begin{definition}\label{Def:DerivBlock}
For input species $E^+, E^-$, auxiliary species $A^+,A^-$, output species $D^+,D^-$, parameters $q,s,v \in\mathbb{R}_{> 0}$ and the multiplier $r \in \mathbb{R}_{\geq0}$, the derivative block is a CRN composed by the following reactions 
\begin{align*}
   E^+ & \overset{rv}{\rightarrow} E^+ + A^+ &
   E^- & \overset{rv}{\rightarrow} E^- + A^-  \\
   E^+ & \overset{rvs}{\rightarrow} E^+ + D^+  &
   E^- & \overset{rvs}{\rightarrow} E^- + D^-  \\
   A^+ & \overset{v}{\rightarrow} \emptyset  &
   A^- & \overset{v}{\rightarrow} \emptyset  \\
   A^+ & \overset{vs}{\rightarrow} A^+ + D^-  &
   A^- & \overset{vs}{\rightarrow} A^- + D^+  \\
   D^+ & \overset{s}{\rightarrow} \emptyset &
   D^- & \overset{s}{\rightarrow} \emptyset &
   D^+ + D^- & \overset{q}{\rightarrow} \emptyset
\end{align*}
\end{definition}

The following theorem shows that the above CRN is such that, under certain scaling of the rates, $(x_{D^+}-x_{D^-}) $ produces a correct approximation of the derivative of $r(x_{E^+} - x_{E^-})$.

\begin{theorem}\label{Derivative}
The derivative blocks are asymptotically correct. In particular, the following holds.
\begin{enumerate}
    \item The solution of~(\ref{eq:fullsys}) converges, on any bounded time interval, to an ODE system which satisfies $x_{D^+_i} - x_{D^-_i} = \partial_t x_{A^+_i} - \partial_t x_{A^-_i}$, for all $1 \leq i \leq n$, if $s \to \infty$ in all derivative blocks.
    \item The solution of~(\ref{eq:fullsys}) converges, on any bounded time interval, to an ODE system which satisfies $x_{A^+_i} = r x_{E^+_i}$ and $x_{A^-_i} = r x_{E^-_i}$ for all $1 \leq i \leq n$, if $v \to \infty$ in all derivative blocks.
\end{enumerate}
\end{theorem}

\begin{proof}[Proof (Sketch)]
To see $1)$, we first note that Definition~\ref{Def:DerivBlock} yields
\begin{align*}
    \partial_t x_{D^+_i} & = svr x_{E^+_i} + vs x_{A^-_i} - s x_{D^+_i} - q x_{D^+_i} x_{D^-_i} \\
    \partial_t x_{D^-_i} & = svr x_{E^-_i} + vs x_{A^+_i} - s x_{D^-_i} - q x_{D^+_i} x_{D^-_i} \\
    \partial_t x_{A^+_i} & = vr x_{E^+_i} - v x_{A^+_i} \\
    \partial_t x_{A^-_i} & = vr x_{E^-_i} - v x_{A^-_i}
    % ODEs without multiplier
    % \partial_t x_{D^+} & = vs x_{E^+} + vs x_{A^-} - s x_{D^+} - q x_{D^+} x_{D^-} \\
    % \partial_t x_{D^-} & = vs x_{E^-} + vs x_{A^+} - s x_{D^-} - q x_{D^+} x_{D^-} \\
    % \partial_t x_{A^+} & = v x_{E^+} - v x_{A^+} \\
    % \partial_t x_{A^-} & = v x_{E^-} - v x_{A^-}
\end{align*}
Since this implies
\begin{multline*}
    \partial_t (x_{D^+_i} - x_{D^-_i}) =  s(vr x_{E^+_i} - v x_{A^+_i}) \\ - s(v r x_{E^-_i} - v x_{A^-_i}) - s (x_{D^+_i} - x_{D^-_i}) ,
\end{multline*}
this motivates us to add to the ODE system of~(\ref{eq:fullsys}) the additional ODE 
\begin{align*}
    \partial_t z_i & = s\underbrace{(v r x_{E^+_i} - v x_{A^+_i})}_{= \partial_t x_{A^+_i}} - s \underbrace{(v r x_{E^-_i} - v x_{A^-_i})}_{= \partial_t x_{A^-_i}} - s z_i
\end{align*}
and to replace each instance of $(x_{D^+_i} - x_{D^-_i})$ in the ODE system with $z_i$. By processing the other derivative blocks in a similar fashion, we introduce new ODE variables $z = (z_1,\ldots, z_n)$. To see that the requirements of Tikhonov's theorem are satisfied, we note that the fast ODE system (i.e., the one containing $z$, $D^+_1,\ldots,D^+_n$ and $D^-_1,\ldots,D^-_n$) admits, for any fixed vector of slow variables, exactly one equilibrium point. Additionally,
it is an asymptotically stable equilibrium of the fast ODE system. To see $2)$, instead, we apply Tikhonov's theorem in the case where, in every derivative block, $x_{A^+_1}, \ldots, x_{A^+_n}$ and $x_{A^-_1}, \ldots, x_{A^-_n}$ are treated as fast variables (while all other variables are considered to be slow) and $v \to \infty$ in all derivative blocks.
\end{proof}

\begin{figure*}[!ht]
\centering
\begin{minipage}[l]{2.0\columnwidth}
\centering
\includegraphics[width=16cm]{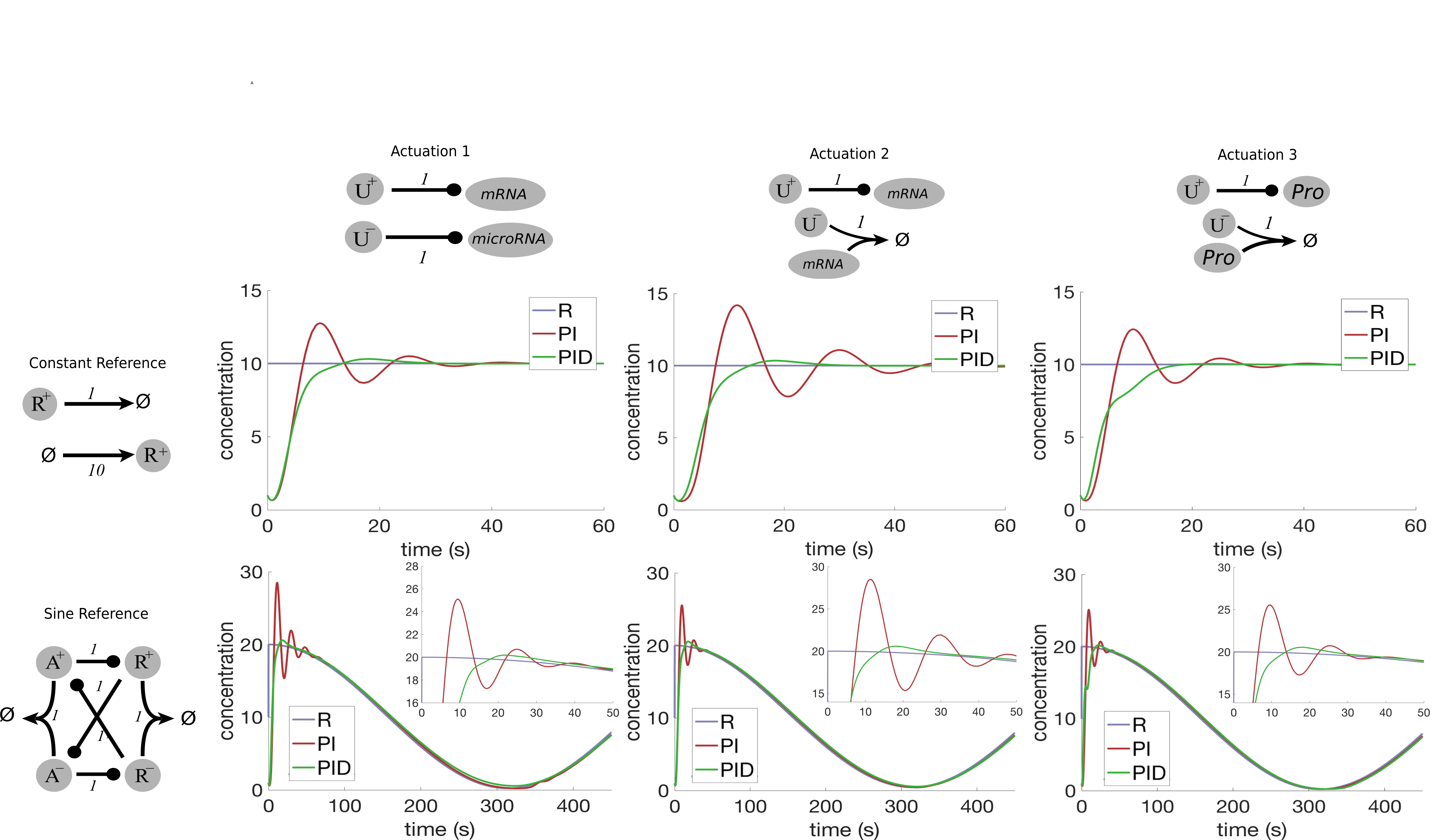}
\caption{%\LL{To be super precise in red and green, in the plots above, we do not plot PI and PID, but the output species for that particular controller. However, I guess this is still clear if we explain it in the caption.} 
We consider the gene expression model given in Section \ref{plantcontrol} and compare the time evolution of the species $\mathit{Pro}$ with different reference signals for PID and PI feedback control with the actuation models described in Equations~\eqref{Eq:Actuation1},~\eqref{Eq:Actuation2}, and~\eqref{Eq:Actuation3} shown in the left, middle, and right column of the plots, respectively. We consider both a constant  and a sine wave reference signal (shown in top and bottom rows of the plots, respectively). 
It is possible to observe that while $\mathit{Pro}$ already tracks correctly the reference signals for PI control, in the case of a PID controller, the output has reduced oscillations around the reference signal. This is emphasised in the insets seen on the sine reference row where we examine the first 50 seconds of the time evolution.}
\label{temp}
\end{minipage}
\end{figure*}

\section{PID control of gene expression}\label{plantcontrol}

In this section we apply the PID feedback control architecture developed in this paper to a gene expression model. In particular, we consider a $\mathit{microRNA}$ regulated gene expression model from~\cite{laurenti2018molecular}, for which synthetic implementations have already been proposed in~\cite{delalez2018design}. The model is composed by the following reactions, where for simplicity we fixed unitary kinetic parameters
\begin{align*}
\emptyset & \overset{1}{\rightarrow} \mathit{mRNA}  & \mathit{mRNA} & \overset{1}{\rightarrow} \emptyset \\
\mathit{mRNA} & \overset{1}{\rightarrow} \mathit{mRNA} + \mathit{Pro} & 
\mathit{Pro} & \overset{1}\rightarrow \emptyset \\
\mathit{mRNA} + \mathit{microRNA} & \overset{1}{\rightarrow} \emptyset &\mathit{microRNA} & \overset{1}\rightarrow  \emptyset \\
\emptyset & \overset{1}{\rightarrow} \mathit{microRNA}. &
\end{align*}
%where $k_1,...,k_7$ are arbitrary reaction rates. \IMT{I think that we should strive for reproducibility of the experiments. I would put the actual kinetic parameters that were used (also for the controllers.}
That is, we have that $\mathit{mRNA}$ catalyses the production of the protein $\mathit{Pro}$ and is down-regulated by an annihilation reaction with the $\mathit{microRNA}$.% $\mathit{mRNA}$ and $\mathit{microRNA}$ are expressed and degraded with same rate. 
%A stochastic version of this model without $\mathit{microRNA}$ has been already studied in \cite{briat2018variance} under integral feedback control. \IMT{So why don't we also study that model under deterministic semantics with PID?} 
%\LL{The problem is that the output of the controller has two components $U^+$ and $U^-.$ $U^+$ represents the positive component of the signal and $U^-$ the negative. As a consequence, we need $U^-$ to act negatively on the output. We can do that either with the microRNA as we do now, or by considering a reaction like $U^- + mRNA \to^k \emptyset ,$ which would give us a model closer to the one studied in \cite{briat2018variance} (we will remove the $microRNA$ from the model in that case). We may then discuss that in a "real world" implementation $U^-$ may be implemented with a microRNA.
%}

The objective of our control is to have the protein $Pro$ to follow a reference signal. Given $U^+$ and $U^-$, the control signals synthesised by the controller, we assume that these can act on the plant by regulating the expression rate of $\mathit{mRNA}$ and $\mathit{microRNA}$, respectively.
This assumption is justified by the fact that these mechanisms can be implemented synthetically  \cite{delalez2018design}.
We consider the following reactions to model such actuation:
\begin{align}
\label{Eq:Actuation1}
    U^+ & \to^{1} U^+ + \mathit{mRNA} & U^- & \to^{1} U^- + \mathit{microRNA}.
\end{align}
In this model, a high concentration of  $U^+$ will increase the production rate of $\mathit{mRNA}$ and so of $\mathit{Pro}$, whereas a high concentration of $U^-$ will decrease the amount of $\mathit{mRNA}$ by producing $\mathit{microRNA}$ with a higher rate.

In the actuation model considered above we have that the control signals act on two different species. 
This is not a requirement of our architecture. Another possible actuation is that $U^-$ annihilates  $\mathit{mRNA}$ directly. This can be modelled with the following reactions
\begin{align}
\label{Eq:Actuation2}
\mathit{U^+} \to^1 \mathit{U^+} + \mathit{mRNA} \quad    \mathit{U^-} + \mathit{mRNA} \to^{1} \emptyset.
\end{align}
Finally, another possibility is that $U^+$ and $U^-$ acts directly on the target species $\mathit{Pro}$. In this case, the actuation is 
\begin{align}
\label{Eq:Actuation3}
\mathit{U^+} \to^1 \mathit{U^+} + \mathit{Pro} \quad    \mathit{U^-} + \mathit{Pro} \to^{1} \emptyset.
\end{align}

%\IMT{What is important here is to mention the parameterisation of the controllers that have been used. Is it the case that the P and I blocks have the same parameters in both the PI and the PID experiments? Are the parameterisations different depending on the reference signals that were used? We need to say something about that here in the main paper even if we give the full model in the appendix (but we should also say that we do so if this is going to be plan.}
In Figure \ref{temp} we consider the different actuation mechanisms described above and compare the performance of PI and PID controllers for two different reference signals: a constant signal and an oscillatory signal.
For all the plots in the figure we considered the same parameters for PI and PID controllers (reported in the Appendix).
It is easy to observe that, whereas a negative feedback with PI controller can already track both signals correctly, in the case of a PID controller the time evolution of the concentration of $\mathit{Pro}$ has reduced oscillations around the reference signals. 
This is due to the action of the derivative block.
In fact, while it is well known that the derivative component in a PID does not necessarily reach zero error at steady state, it can help to reduce the transient error between the output and the reference signals and to dampen oscillations around the set points.

\section{Conclusion}
In this work we considered feedback control with PID controllers and proposed a CRN implementation for this control architecture. This relies on a novel CRN, which computes the derivative of an input molecular signal.  
We applied our framework to control the protein expression in a $\mathit{microRNA}$ regulated gene expression model and showed improved performance compared to a PI feedback control. 
An interesting aspect, which has not been considered in this paper, is to study the effect that the proposed control system has on noise \cite{briat2018antithetic}. This is left as future work.

%\LL{MaybeIn Appendix put also the stochastic simulations }
%The main contribution of this paper is a novel CRN design of a derivative action, which we validate by comparing a PI controller, composed of the proportional and integral actions, against a PID controller, with the added derivative action. In Figure \ref{Figure:option1} we provide simulations of both the PI and PID controller against constant and sine waves generated by the CRNs listed above the diagram. In both cases rate parameters are fixed in the PI blocks of the controller to highlight the result of adding the derivative action. In all cases the output of the plant is given by the species $P$ and the reference signal is given by the species $R$.

%In the case of the constant signal, where there is no derivative action, which contributes to the stability, we  observe oscillation around the reference signal and a longer convergence time of 60 seconds as opposed to the PID controller which converges in 30 seconds. In the case of the sine wave reference signal we observe that the plant output oscillates around the inflection point. In contrast there is no such oscillation with our derivative action introduced in the full PID controller which tracks the reference signal. The convergence time is twice as fast at 100 seconds in the PID controller as opposed to 200 seconds in the PI controller.

\bibliographystyle{unsrt}
\bibliography{biblio}

\newpage

\appendix

\subsection{CRN for PID control of gene expression}
We present our full Chemical Reaction Network PID feedback loop with gene expression plant and two reference signals introduced in Section \ref{plantcontrol}. We include the three actuation mechanisms given with the plant. For each block we also give the initial conditions used to produce the simulations seen in Figure 3. %The following reaction network has been simulated using \textit{Microsoft's Visual GEC} \cite{}. 

\subsection*{PID Controller}

First we report the reactions and parameters of the CRN PID controller of which the proof of correctness and details of operation are outlined in Section 3. For all figures we used the same parameters \\

\textbf{Proportional Block:}
\begin{align*}
    &E^+  \overset{rs}{\rightarrow} E^+ + P^+  
    &&E^-  \overset{rs}{\rightarrow} E^- + P^- \\
    &P^+ + P^-  \overset{q}{\rightarrow}\emptyset
    &&P^+  \overset{s}{\rightarrow}\emptyset \\
    &P^-  \overset{s}{\rightarrow}\emptyset
\end{align*} 

Initial Conditions: (rates) $s,rs,q = 1$ (species) $P^+$,$P^- = 0$\\

\textbf{Integral Block:}
\begin{align*}
   & E^+ \overset{k}{\rightarrow} E^+ + I^+ &
   & E^- \overset{k}{\rightarrow} E^- + I^- &
   & I^+ + I^-\overset{q}{\rightarrow}\emptyset
\end{align*}

Initial Conditions: (rates) $k,q =1$, (species) $ I^+$, $I^- = 0$\\

\textbf{Derivative Block:}
\begin{align*}
   E^+ & \overset{rv}{\rightarrow} E^+ + A^+ &
   E^- & \overset{rv}{\rightarrow} E^- + A^-  \\
   E^+ & \overset{rvs}{\rightarrow} E^+ + D^+  &
   E^- & \overset{rvs}{\rightarrow} E^- + D^-  \\
   A^+ & \overset{v}{\rightarrow} \emptyset  &
   A^- & \overset{v}{\rightarrow} \emptyset  \\
   A^+ & \overset{vs}{\rightarrow} A^+ + D^-  &
   A^- & \overset{vs}{\rightarrow} A^- + D^+  \\
   D^+ & \overset{s}{\rightarrow} \emptyset &
   D^- & \overset{s}{\rightarrow} \emptyset &
   D^+ + D^- & \overset{q}{\rightarrow} \emptyset
\end{align*}

Initial Conditions: (rates) $v,s = 1$, $vs,q,rvs = 10$ , (species) $A^+$,$A^-$,$D^+$,$D^- = 0$\\

\subsection{Summation and Subtraction Blocks}

We provide the CRNs for the two summation blocks which are highlighted in the green block in Figure 2 and noted upon in section 3. The first adds the output of the P and I blocks together. The second the PI and D blocks together which produces the output species of the controller $U$ (given as $PID$ in the CRNs below). \\ 

\textbf{Addition Block $P + I$:}

\begin{align*}
   P^+ & \overset{s}{\rightarrow} P^+ + B^+ &
   I^+ & \overset{s}{\rightarrow} I^+ + B^+  \\
   P^- & \overset{s}{\rightarrow} P^- + B^-  &
   I^- & \overset{s}{\rightarrow} I^- + B^-  \\
   B^+ + B^- & \overset{q}{\rightarrow} \emptyset  &
   B^+ & \overset{s}{\rightarrow} \emptyset  \\
   B^- & \overset{s}{\rightarrow} \emptyset  
\end{align*}

Initial Conditions for $P + I$ summation block:
(rates) $s=0.8$,$q=0.3$, (species)  $B^+$,$B^- = 0$\\

\textbf{Addition Block $PI + D$:}

\begin{align*}
   PI^+ & \overset{s}{\rightarrow} PI^+ + PID^+ &
   D^+ & \overset{s}{\rightarrow} D^+ + PID^+  \\
   PI^- & \overset{s}{\rightarrow} PI^- + PID^-  &
   D^- & \overset{s}{\rightarrow} D^- + PID^-  \\
   PID^+ + PID^- & \overset{q}{\rightarrow} \emptyset  &
   PID^+ & \overset{s}{\rightarrow} \emptyset  \\
   PID^- & \overset{s}{\rightarrow} \emptyset  \\
\end{align*}

Initial Paramterisation for $PI + D$ summation block: (rates) $s=1.1$,$q=0.1$ (species) $PID^+$,$PID^- = 0$\\

\textbf{Subtraction Block:}

The subtraction block is used to compute the error of the output of the plant $Y$ with the reference signal. It is detailed further in Section 3:

\begin{align*}
   Y'^+ & \overset{s}{\rightarrow} Y'^+ + E^+ &
   R^- & \overset{s}{\rightarrow} R^- + E^+  \\
   Y'^- & \overset{s}{\rightarrow} Y'^- + E^-  &
   R^+ & \overset{s}{\rightarrow} R^+ + E^-  \\
   E^+ + E^- & \overset{q}{\rightarrow} \emptyset  &
   E^+ & \overset{s}{\rightarrow} \emptyset  \\
   E^- & \overset{s}{\rightarrow} \emptyset 
\end{align*}

Initial Conditions for difference block summation block:
(rates) $s,q = 1$\\

\subsection{Reference Signals}

Next we introduce the intitial conditions and reactions for both the constant and sine wave reference signals outlined further in Section IV. These act as a reference which we are trying to control our plant to track. \\

\textbf{Constant:}

The constant signal can be given simply by stating a non-decaying species with a molecular count equal to the constant signal however it can also be given by the following CRN which is stated in Figure 3. 

\begin{align*}
   \emptyset & \overset{k}{\rightarrow} R^+
   &R^+ \overset{r}{\rightarrow} \emptyset  \\
\end{align*}

Initial Conditions: (rates) $k = 10, r = 1$ (species) $R^+ = 0$, $R^- = 0$ \\

\textbf{Sine wave:} 
The sine wave has a slow reaction rate to allow for the PID controller to properly track the signal.
\begin{align*}
   A^+ & \overset{k}{\rightarrow} A^+ + R^+ &
   A^- & \overset{k}{\rightarrow} A^- + R^-  \\
   R^+ & \overset{k}{\rightarrow} R^+ + A^-  &
   R^- & \overset{k}{\rightarrow} R^- + A^+  \\
   A^+ + A^- & \overset{k}{\rightarrow} \emptyset &
   R^+ + R^- & \overset{k}{\rightarrow} \emptyset 
\end{align*}

Initial Conditions: (rates) $k = 0.01$ (species) $A^+ = 10$, $A^- = 0$, $R^+ = 0$, $R^- = 0$\\

\subsection*{Plant and Actuators}

We introduce the gene expression plant used within our model. We also include the three actuations methods from the controller $U^+, U^-$ seen in Figure 3 and discussed in section 4 used to interface with the plant. \\

\textbf{Actuator 1:}
\begin{align*}
   U^+ &  \overset{k}{\rightarrow} U^+ + mRNA &
   U^- & \overset{k}{\rightarrow} U^- + microRNA
\end{align*}

Initial Conditions: 
(rates) $k=1$, $mRNA=0$, $microRNA=0$ 
(species) $U^+ = 0.5$, $U^- = 0.5$ \\

\textbf{Actuator 2:}
\begin{align*}
   U^+ &  \overset{k}{\rightarrow} U^+ + mRNA &
   U^- + mRNA & \overset{k}{\rightarrow} U^-
\end{align*}

Initial Conditions: (rates) $k=1$, $mRNA=0$ (species) $ U^+ = 0.5$, $U^- 0.5$ \\

\textbf{Actuator 3:}
\begin{align*}
   U^+ &  \overset{k}{\rightarrow} U^+ + Pro &
   U^- + Pro & \overset{k}{\rightarrow} U^-
\end{align*}

Initial Conditions: (rates) $k=1$, $Pro=1$ (species) $U^+ = 0.5$, $U^- 0.5$ \\

\textbf{Plant:
}
\begin{align*}
   & \overset{k}{\rightarrow} mRNA &
   mRNA & \overset{k}{\rightarrow} \emptyset  \\
   mRNA^+ & \overset{k}{\rightarrow} mRNA + Pro &
   Pro & \overset{k}{\rightarrow} \emptyset \\
   && mRNA + microRNA &\overset{k}{\rightarrow} \emptyset  &\\
\end{align*}

Initial Conditions:  (rates) $k=1$, 
(species) $mRNA$, $microRNA = 0$, $Pro = 1$ \\

\subsection{Single to Dual Rail Block}

We introduce the block which takes the output of the plant $Y$ and transforms into into a dual rail signal, see blue block Figure 2 and discussed in section 3.

\begin{align*}
   Y & \overset{s}{\rightarrow} Y + Y'^+ &
   Y'^+ & \overset{s}{\rightarrow} \emptyset &
   Y'^+ + Y'^- & \overset{q}{\rightarrow} \emptyset   
\end{align*}

Initial Conditions: (rates) $s,q = 1$, (species) $Y^{'+}$,$Y^{'-} = 0$

\end{document}